\documentclass[a4paper,UKenglish,cleveref, autoref, thm-restate]{lipics-v2021}

\usepackage[ruled,vlined]{algorithm2e}

\SetCommentSty{mycommfont}

\newcommand{\abs}[1]{\left|#1\right|}

\newcommand{\E}[1]{\mathbb{E}\left[#1\right]}

\EventEditors{Kunal Talwar}
\EventNoEds{1}
\EventLongTitle{4th Symposium on Foundations of Responsible Computing (FORC 2023)}
\EventShortTitle{FORC 2023}
\EventAcronym{FORC}
\EventYear{2023}
\EventDate{June 7--9, 2023}
\EventLocation{Stanford University, California, USA}
\EventLogo{}
\SeriesVolume{256}
\ArticleNo{3}
\title{Bidding Strategies for Proportional Representation in Advertisement Campaigns}

\author{Inbal Livni Navon}{Stanford University, USA}{inballn@stanford.edu}{https://orcid.org/0000-0001-5949-316X}{Supported by the Sloan Foundation Grant 2020-13941 and the Zuckerman STEM Leadership Program.}

\author{Charlotte {Peale}}{Stanford University, United States}{cpeale@stanford.edu}{https://orcid.org/
0000-0002-9959-857X}{Supported by the Simons Foundation Collaboration on the Theory of Algorithmic Fairness.}

\author{Omer {Reingold}}{Stanford University, United States}{reingold@cs.stanford.edu}{https://orcid.org/
0000-0003-4997-1716}{Supported by the Simons Foundation Investigators Award 689988.}

\author{Judy Hanwen {Shen}}{Stanford University, United States}{jhshen@cs.stanford.edu}{https://orcid.org/0000-0002-7864-5242}{Supported by the Simons Foundation Collaboration on the Theory of Algorithmic Fairness.}

\authorrunning{I. Livni Navon et al. } 

\Copyright{Inbal Livni Navon, Charlotte Peale, Omer Reingold, Judy Hanwen Shen} 



\category{} 

\relatedversion{} 

\acknowledgements{}

\nolinenumbers 
\date{}

\begin{document}
\maketitle
\begin{abstract}
Many companies rely on advertising platforms such as Google, Facebook, or Instagram to recruit a large and diverse applicant pool for job openings. Prior works have shown that equitable bidding may not result in equitable outcomes due to heterogeneous levels of competition for different types of individuals. Suggestions have been made to address this problem via revisions to the advertising platform. However, it may be challenging to convince platforms to undergo a costly re-vamp of their system, and in addition it might not offer the flexibility necessary to capture the many types of fairness notions and other constraints that advertisers would like to ensure. Instead, we consider alterations that make no change to the platform mechanism and instead change the bidding strategies used by advertisers. We compare two natural fairness objectives: one in which the advertisers must treat groups equally when bidding in order to achieve a yield with group-parity guarantees, and another in which the bids are not constrained and only the yield must satisfy parity constraints. We show that requiring parity with respect to both bids and yield can result in an arbitrarily large decrease in efficiency compared to requiring equal yield proportions alone. We find that autobidding is a natural way to realize this latter objective and show how existing work in this area can be extended to provide efficient bidding strategies that provide high utility while satisfying group parity constraints as well as deterministic and randomized rounding techniques to uphold these guarantees. Finally, we demonstrate the effectiveness of our proposed solutions on data adapted from a real-world employment dataset.   
\end{abstract}

\pagebreak
\section{Introduction}

For many institutions, hiring a diverse workforce is crucial to achieving and retaining an equitable environment. While there are many strategies that can be employed to ensure that each stage of the hiring process, from initial resume screening to a final hiring decision, can be done equitably, even the best of attempts may fall short if the initial pool of applicants lacks sufficient diversity. As a result, many companies rely on online advertising platforms such as Google, Facebook, or Instagram to recruit a wider applicant pool for job openings. 

Advertising platforms sell slots to advertisers through auction mechanisms. Toward the goal of yielding a diverse applicant pool, advertisers are able to create recruitment and marketing campaigns to target users of different demographic groups. This specific but salient setting of job advertisements is bound by policy oversight from different government entities. In the United States, the Equal Employment Opportunity Commission enforces discrimination laws that prohibit employers from ``\textit{publishing a job advertisement that shows a preference for or discourages someone from applying for a job based on his or her race, color, religion, sex, national origin, age, disability or genetic information}"\footnote{\url{https://www.eeoc.gov/prohibited-employment-policiespractices}}. However, it is unclear whether this guidance refers to bidding equally on individuals from different demographics, or to achieving a proportional yield for all demographics regardless of protected class status. 

Prior work has observed that these two goals may not be equivalent. Due to differences in advertiser demand, the required costs to reach users of various demographic groups can be very different; women in particular may see fewer job ads due to competition from retail brands that do not target men \cite{lambrecht2019algorithmic}. As a result, setting the same bid value for all groups may still result in a disproportionate representation in downstream yield when there are different levels of competition for different groups of users on the platform \cite{dwork2018fairness, chawla2019multi}. 
Existing work (such as~\cite{chawla2019multi, chawla2022individually} with more discussed in Section~\ref{sec:lit-rev}) interprets this behavior as a failure of the mechanism due to composition, and suggests ways that ad auctions could be redesigned to guarantee fair outcomes despite these composition effects. 

However, an advertising platform may be unlikely to implement a new auction mechanism for a number of reasons, even if the revenue of the new alternative can be shown to be close to that of the original mechanism. For instance, the costs necessary to research, deploy, and completely redesign the current auction system may make such a change undesirable. Moreover, a new auction might make the mechanism far more complex and difficult for advertisers to understand as well as offer less flexibility if it is designed with only a few specific types of constraints and objectives in mind. We discuss these concerns in more detail in Section~\ref{sec:arguments}.

Instead, we take the perspective that perhaps only requiring advertisers to bid values that are similar across different groups of interest may not be the most useful requirement for this context if the fairness of the system is judged by the outcomes of the auctions, and not the bids that are inputted. In fact, in Section~\ref{sec:utility-gaps}, we demonstrate that requiring advertiser's bids to be similar across groups may actually \emph{hinder} achieving parity with respect to auction outcomes, and show that the utility of the optimal bidding strategy that satisfies parity constraints at both the bid and outcome level can be far lower than the optimal bidding strategy that requires group parity at the outcome level alone.

We use these arguments and examples to motivate an alternative approach to redesigning auctions, which is to consider the perspective of an individual advertiser and design bidding strategies that guarantee outcomes that meet the advertiser's goals. This approach is often referred to as an ``autobidder,'' and the adoption of such technologies as a way to control spending and budget depletion is growing increasingly popular. There are a few works that consider autobidding for group parity goals~\cite{celli2022parity, nasr2020bidding}, but these approaches do not give formal guarantees about how closely the resulting bidding strategies meet the desired parity constraints.

We argue that due to their flexibility, practicality, and ease of implementation in existing systems, autobidding strategies that guarantee proportional representation across key subgroups are a key direction for research in equitable online ads. In this paper, we show how we can build on the autobidding framework of Aggarwal et al.~\cite{aggarwal2019autobidding} to develop an efficient algorithm to compute bidding strategies with provable proportional group representation guarantees in the offline setting. We additionally show how our constraints fit into the model studied by Castiglioni et al.~\cite{castiglioni2022unifying} to provide efficient online bidding algorithms with sublinear regret. We focus on strategies for a single autobidder, though understanding market dynamics when many autobidders with fairness constraints are deployed is a natural next step for future work. 

We supplement our arguments and algorithms with empirical evidence using data modified from the American Community Survey. By comparing single bid, gender-based bids, and our autobidder, we see that the autobidder achieves the best combination of utility and representation across different job sectors with different levels of representation.

\subsection{Contributions}
    \paragraph*{Motivations and Fairness Notions} Through examples and qualitative analysis, we consider different potential notions of fairness and group representation for ad auctions and make the case for using autobidders to achieve equitable ad exposure (Sections~\ref{sec:arguments} and \ref{sec:model}). 
    \paragraph*{Optimal Randomized Bidding Strategy for Budget and Group Representation Constraints} Building on the autobidder with constraints framework suggested in \cite{aggarwal2019autobidding}, we add constraints on the representation of key subgroups in the set of clicks resulting from a series of auctions and show that it is possible to calculate an approximately optimal bidding strategy that stays within budget while satisfying group representation constraints in expectation. We consider two platform revenue schemes:  one in which advertisers only pay for clicks and another where advertisers pay for impressions (Section~\ref{sec:offline-alg}).

    \paragraph*{Bidding Strategy with Deterministic Constraint Guarantees} Our model assumes that each individual $i$ clicks on an ad with some probability $ctr_i$, therefore there is inherent randomness in the outcome, and it is not possible to have a deterministic promise on any of the constraints. However, we give a modification to the algorithm that results in slightly lower utility, but satisfies the constraints with high probability, and not only in expectation. The randomized rounding method assumes that there exists a large fraction of the population without representational constraints (Section~\ref{sec:rounding}). 

    \paragraph*{Rounding for Deterministic Solutions When Groups are Disjoint} In the special case of disjoint groups and constraints on group representation with respect to impressions (rather than clicks), we show how to achieve a deterministic solution with utility that is close to optimal. This solution works also for a small number of intersection groups.
    For intersecting groups, it is possible to use the randomized rounding described above, that promises the constraints are met with high probability (Section~\ref{sec:rounding}). 

    \paragraph*{Extension to Autobidding with Representation Constraints in the Online Setting} We show that our constraints can be satisfied by an online algorithm with sublinear regret using the online learning framework of Castiglioni et al.~\cite{castiglioni2022unifying} (Section~\ref{sec:online-ext}).
    \paragraph*{Empirical Data on Autobidder Performance} Using data adapted from the American Community Survey and the US Bureau of Labor Statistics, we simulate the results of different bidding strategies. We show the advantage of our proposed autobidder with proportional representation constraints and randomized rounding for achieving both equitable exposure and high advertiser utility (Section~\ref{sec:experiments}). 

\section{Related Work}\label{sec:lit-rev}
\subsection{Mechanism Design}

A number of existing works consider ways to design auctions that satisfy different choices of fairness guarantees. Chawla et al.~\cite{chawla_et_al:LIPIcs.ITCS.2022.42} design truthful auctions that guarantee individually fair outcomes, Celis et al. \cite{celis2019toward} incorporate group parity constraints into an auction mechanism, and Kuo et al.~\cite{kuo2020proportionnet} propose a deep learning approach to approximately optimal auctions while incorporating relaxed individual fairness constraints. Dwork et al.~\cite{dwork2018fairness} observe that even when advertisers place bids that fulfill their personal fairness goals, competition from other advertisers may prevent the auction outcomes from satisfying advertisers’ fairness constraints. Building on this observation, subsequent works design auction mechanisms whose outcomes satisfy individual fairness constraints for each advertiser, assuming that advertisers’ bids satisfy individual fairness guarantees with respect to their personal metrics \cite{chawla2019multi, chawla2022individually}.

\subsection{Autobidding Strategies}

A different approach for achieving advertising auction goals is to consider the problem from the point of view of the advertisers (bidders) and design bidding strategies that guarantee the desired properties. This approach is often termed ``autobidding''.  While a number of works explore autobidding strategies for a number of different budgets and spending-related goals~\cite{balseiro2019learning, borgs2007dynamics, conitzerMultiplicative2022, lucier2023autobidders, aggarwal2019autobidding}, autobidding strategies that optimize for fairness guarantees are still relatively under explored. Nasr et al.~\cite{nasr2020bidding} first suggest adding parity constraints specifically to bidding strategies using an unlimited budget. Celli et al.~\cite{celli2022parity} explore autobidding strategies that incorporate parity constraints via a regularization term in the objective function. However, this approach does not allow for any formal guarantees about how well these fairness goals are achieved by the algorithm. 
More recently, Castiglioni et al.~\cite{castiglioni2022unifying} consider more general autobidding strategies that can handle a number of different types of constraints. While they do not consider fairness guarantees in their main results, they note that fairness constraints could be considered as a future direction and suggest one potential formulation of fairness constraints. 
Unfortunately, the fairness constraints they suggest are not proven to yield an efficient algorithm because the solutions are not guaranteed to be feasible. We show that our constraints do satisfy feasibility requirements and yield an efficient online autobidding algorithm.

Another related line of work considers how to "learn to bid" or how to discover bidding strategies using feedback from the outcomes of repeated auctions~\cite{ltb-balseiro2019contextual, ltb-feng18, ltb-han2020learning, ltb-han2020optimal, ltb-varadaraja23, ltb-weed16}. Our work mostly considers the offline full information setting in which the winning bids, items, and values are all known to the autobidding algorithm. In Section \ref{sec:online-ext}, we note that the online learning framework developed by Castiglioni et. al in \cite{castiglioni2022unifying} can be extended to give online bidding algorithms that guarantee group proportionality constraints will be satisfied in the long-term when slots and their associated values are drawn from a stationary distribution at each step. However, further exploring how to learn fairness-aware bidding strategies in the repeated auction setting is an interesting direction for future work.

\section{Motivation}\label{sec:arguments}
We consider a two-part system, in which advertisers place bids on individual ad slots according to some bidding algorithm and a centralized auction mechanism decides which advertiser gets a particular slot, and how much they will pay. In reality, platforms like Google will often perform both parts of this process once an advertiser defines a campaign with a target population, total budget ($B$), and potentially some additional desiderata. Currently, most ad platforms use a standard second- or first-price auction to decide how ads are allocated, though there have been proposals for alternative options that guarantee a variety of different fairness objectives (See Section~\ref{sec:lit-rev}).

We observe that there are a number of different reasons why it makes sense to focus on designing a new bidding algorithm rather than implementing alterations to the auction mechanism itself when the system contains bidders who would like to ensure their ads result in an even spread of clicks from their target population. 
\paragraph*{Cost to Platform} Most alternative options do not consider the potential loss in revenue for the platform that would arise from implementing a new auction mechanism. Even when an alternative mechanism can provide near-optimal platform utility, significant costs associated with designing, implementing, and switching over to a new mechanism are likely to make such a switch impractical from the point of view of a platform. 
\paragraph*{Loss of Flexibility} There are many different objectives and constraints that advertisers would like to optimize for. While we focus on group representation constraints, some advertisers may be more focused on other objectives such as alternative notions of fairness, or goals outside fairness such as a limit on their rate of spending. Keeping the auction as a fixed mechanism allows advertisers to specify their own individual constraints and optimize their bids to match.  
\paragraph*{Decreased Comprehension} Ad platforms prioritize simplicity in their auction mechanisms. For this reason, many companies including Google have recently decided to switch from second- to first-price auctions, citing concerns around simplifying the ad-purchasing process for advertisers~\cite{adsenseFAQ}. It, therefore, seems unrealistic to expect platforms to switch to the more complicated mechanisms required to enforce fairness guarantees.

Due to these reasons, we concentrate on the question of designing optimal bidding algorithms for advertisers with group representation constraints.  

\subsection{What Does it Mean to Bid Fairly in Second and First-Price Auctions?}\label{sec:utility-gaps}

Prior works have observed a composition problem that arises in standard auction settings~\cite{dwork2018fairness}. As a simple example, we assume that some advertiser values individuals from groups $A$ and $B$ equally, and so bids the same value $v$ on individuals from each group. 

In a vacuum, such a strategy would result in having a proportional number of ads shown to both groups. However, other advertisers in the market may not have the same goals, and may specifically target one group by only bidding on individuals in group $A$. When composed together, the many bidding strategies used by all the different advertisers in the market may result in different winning bid values for the two groups. In particular, an individual from group $B$ may require a winning bid of $v$, but individuals from $A$ may require a higher bid of $2v$ due to increased demand. In this situation, our advertiser's strategy will result in ads shown only to group $B$, rather than to both groups proportionally. 

The perspective of existing work is that in this example, our advertiser was ``doing the right thing'', i.e. bidding on groups similarly, and it is a failure of the composition mechanism (the auction) that causes differential rates of ad exposure. Instead, we argue that bidding in a way that satisfies group parity constraints might not be the right notion for this context given that practical goals of diverse recruitment are judged based on the auction outcomes. In fact, there are three potential general types of fairness that could be considered here, defined by different parts of the bidding process. We discuss these three options below in terms of group parity guarantees. However, this framing applies to other notions such as individual fairness as well. 
\begin{enumerate}
    \item (\emph{Bid Parity}) The advertiser is required to bid similarly across different groups of interest. This is the notion we considered in our above example, where we saw that bid parity alone does not guarantee that yields will satisfy any sort of group proportionality goals. 
    \item (\emph{Outcome Parity}) Instead, we could explicitly require that an advertiser's yield (measured either in terms of clicks or exposures, depending on the setting) has representation of key groups that is proportional to their representation in the population. Here, we do not put any constraint on how advertisers must bid to achieve a proportional yield. 
    \item (\emph{Bid-and-Outcome Parity}) Lastly, we could potentially consider a stricter notion that requires that both an advertiser's bids be similar across groups \emph{and} the resulting yields be proportional. 
\end{enumerate}
If we care about the outcomes of ad auctions, it's natural to focus on either outcome parity or bid-and-outcome parity as goals for a bidding algorithm. On first glance, these might seem somewhat similar. Clearly, any strategy satisfying bid-and-outcome parity will also satisfy outcome parity, however, we can show that the opposite direction does not necessarily hold. In fact, a simple example demonstrates that strategies satisfying bid-and-outcome parity may result in arbitrarily large decreases in advertiser utility compared to strategies that are only required to satisfy outcome parity in both second and first-price auctions. 
\begin{example}
    We consider a second-price auction\footnote{By similar reasoning, it's easily verified that a first-price auction run in the same setting would result in even larger gaps in utility, so we concentrate on second-price auctions for this example. We also only focus on yield in terms of exposure here for simplicity, but the example can be easily extended to work for yield that is measured in terms of clicks as well by incorporating click-through-rates.} being run on a population partitioned into two groups, $A$ and $B$. 

    We suppose that an advertiser has a budget of $\$5$ that it uses to bid on a population of 100 individuals $(G, w) \in P \subseteq \{A, B\} \times W$, where G corresponds to an individual's group, and $w$ corresponds to the winning bid from a discrete set of bids $W = \{\$0.1, \$0.4, \$1\}$ (if an advertiser bids $b \geq w$, they win the auction and pay $w$, and do not win the auction and pay nothing otherwise). 

    Groups are distributed evenly across the population, so there are 50 individuals from group A and 50 individuals from group B. However, the distributions of winning bids are skewed slightly to the right (higher cost) for individuals from group $A$ compared to group $B$, i.e. we have the following numbers of individuals with each winning bid:

        \begin{center}
        \begin{tabular}{ |c|c|c|c| } 
         \hline
          & w = \$0.1 & w = \$0.4 & w = \$1 \\ 
          \hline
         Group A & 25 & 20 & 5 \\ 
         \hline
         Group B & 40 & 10 & 0 \\ 
         \hline
        \end{tabular}
        \end{center}

    We consider an offline setting where these winning bids and numbers of individuals are all known to an advertiser beforehand and used to set a bidding strategy, and then these 100 individuals arrive in a random order and the bidding strategy is applied until the budget runs out. We consider two options for bidding strategies. First, a bid-constrained strategy is one where an advertiser must set a maximum bid $b$, and bid $b$ on every individual that arrives until the budget runs out (this translates to bidding $b$ on every individual with equal probability because the order is randomized). In Appendix~\ref{sec:bidfairness}, we discuss why this is a natural definition of bid parity in this setting. The second option is to use a bid-unconstrained strategy. In this approach, an advertiser can set a unique bid for each individual.
    
    We assume that an advertiser values all individuals equally, and thus its utility is equal to the number of individuals that are shown an ad. When the outcome is required to be proportional to the group sizes, an advertiser must bid in such a way that the expected number of ads shown to group $A$ is equal to the number of ads shown to group $B$.

    When bids are unconstrained and the advertiser can decide the bid amount for each individual separately as long as the outcomes are proportional to group sizes, it's optimal for an advertiser to bid $\$0.1$ on 25 individuals with $w = \$0.1$ from group $A$ and 25 individuals with $w = \$0.1$ from group $B$, and bid $w=\$0$ on every other individual. This results in ads shown to 50 individuals total, which is equal to the optimal number that could be reached even when outcomes are unconstrained. 

    In contrast, when an advertiser must use a bid-constrained strategy, setting the maximum bid $b$ to be any value smaller than $\$1$ cannot satisfy group proportionality constraints because the expected number of individuals shown ads from group $B$ will always be larger than for group $A$. Thus, the only strategy that satisfies both bid and outcome parity is to bid the maximum-possible bid of $\$1$ on all individuals until the budget runs out. 

    This results in a strategy that shows ads to only $21.3$ individuals in expectation, less than half the utility of the bid-unconstrained strategy. Moreover, note that this strategy provides the lowest utility of \emph{any} of the potential bid-constrained strategies.

\end{example}
    This example exhibits a setting in which requiring bid parity in addition to outcome parity may result in much lower utility for advertisers. We note that this example can be extended to even larger spreads of price distributions where the distribution of group A is slightly skewed right in comparison to group B, again requiring a bid-constrained strategy to bid the maximum possible winning bid of any individual to receive proportional outcomes, whereas a bid unconstrained strategy can satisfy outcome parity while matching the utility of the optimal unconstrained bidding strategy. 

    We conclude that requiring advertisers to bid in a way that respects parity constraints does not directly contribute to receiving group-proportional outcomes, and in some situations may actually make achieving such outcomes incredibly costly compared to strategies where bids are unconstrained. This motivates our interest in optimal bidding strategies that satisfy outcome parity, which we explore in the following sections.
\section{Autobidder with Constraints on Subgroup Representation}\label{sec:offline-alg}

Now that we have justified our perspective and proposed approach, we describe how we choose to model an ad auction from an advertiser's point of view, and how to compute optimal bidding strategies for this setting.

\subsection{Setup}\label{sec:model}

We consider a large set of queries (or individuals) $I$, each of which has a single slot that can show an ad. For each query $i$, an auction determines which ad is shown as well as the cost-per-click ($cpc_i$) of the ad. 

We consider a static setting in which we are trying to set the bid of single advertiser with full knowledge of the bids of other advertisers, i.e. there is a set $cpc_i$ for each query, and the advertiser wins the ad if and only if their bid is above that value. In a first-price auction, the winning bidder pays their bid. In a second-price auction, the winning bidder will pay $cpc_i$. This is a practical assumption in larger markets since $cpc_i$ remains stable. Because our model assumes that we know $cpc_i$ for each individual, the optimal strategies for first-price and second-price auctions are equivalent, because there is no need to bid higher than whatever would be paid in a second-price auction.  

Each query also has an associated click-through-rate ($ctr_i \in [0, 1]$) and value to the advertiser: $v_i \geq 0$. A bidder's goal is to select the optimal set of queries $I^*$ that maximize its expected value $\sum_{i \in I^*} v_i ctr_i$, subject to a set of budget constraints and representation constraints. Budget constraints ensure that the advertiser's expected cost stays below some threshold. Here we will focus on the simplest type of budget constraint that just requires the total expected cost is within a budget $B$: $\sum_{i \in I^*}ctr_icpc_i \leq B$. However, our approach can be extended to more complicated sets of budgetary constraints. 

The second type of constraint we consider is a \emph{group representation constraint}, which allows the advertiser to ensure that the clicks it receives contain sufficient representation from key demographic groups. We allow an advertiser to specify its goal via a set of constraints that require the proportion of clicks from a particular group $g \subseteq I$ to be at least some goal value $\mu_g$, i.e. 
$ \sum_{i \in I^* \cap g}ctr_i \geq \mu_g \sum_{i \in I^*}ctr_i$.

\subsection{Optimal Ad Allocation as a Linear Program}
We express the search for an optimal $I^*$ as described above as an integer linear program, in which the variables $x_i$ correspond to whether or not the advertiser should win the auction for the $i$th slot. We assume that the advertiser's spending is limited by a budget $B$, and we are given a set of groups $G$, where each $g \in G$ is associated with a lower bound on the desired fraction of total clicks that come from group $g$, $\mu_g \in [0, 1]$. For a group $g$, denote $g_i:= \mathbf{1}[i \in g]$ as a binary indicator variable for query $i$'s membership in $g$.
\begin{align}\label{lp:primal}
\begin{array}{ll}
\mbox{maximize} &  \sum_{i}x_ictr_iv_i \\
\mbox{subject to} & \sum_i x_ictr_icpc_i \leq B \\
        & \sum_i x_i ctr_i (\mu_g - g_i) \leq 0,~ \forall g \in G \\
        & x_i \in \{0, 1\}, ~\forall i \in I.
\end{array}
\end{align}
We can relax the above program by allowing $0\leq x_i\leq 1$, where fractional $x_i$s can represent the probability the advertiser should win the auction for slot $i$. We denote it as the relaxed ad allocation linear program.

\begin{theorem}\label{thm:agl}
    Let $P$ be a relaxed ad allocation linear program. Let $\mathcal{V}$ be an bound on the objective value, and for each constraint $c$, let $V_c$ be an upper bound on the violation of constraint $c$. Then for every $\delta>0$, \Cref{alg:solve} outputs a solution $x\in[0,1]^n$ with utility within $\delta\mathcal{V}$ of the optimal utility achievable by the relaxed linear program and violates the each of the constraints with up to $\delta V_C$ additive error. 
\end{theorem}

On \Cref{lem:rand-round} we show that under certain conditions it is possible to have a randomized rounding algorithm satisfying all the constraints with high probability, and in \Cref{lem:det-round} we show that for disjoint groups, there is a deterministic rounding algorithm satisfying the constraints with a small additive error.

We prove the theorem by adapting the multiplicative weights algorithm presented in \cite{aggarwal2019autobidding}, where it was used to solve a linear program with only budgetary constraints. We show that this approach can be modified to work for our setting as well. 

At a high level, the algorithm from \cite{aggarwal2019autobidding} assumes some known rough bounds on the maximal objective value of the linear program, $\mathcal{V}$, and rough bounds on the amount of violation of each constraint. It then searches for the optimal objective by considering candidate objective values $V$ and for each $V$, searching for a solution whose objective value is equal to $V$. The search is done by a multiplicative weights algorithm that solves a series of one-dimensional problems. In this setting, the solution for each of these one-dimensional problems has a closed form in terms of a thresholds $T_i$. The algorithm runs in time $O(n^2/\delta^4|G|)$ to get a $\delta \mathcal{V}$-approximate solution, where $\mathcal{V}$ is a bound on the maximal possible utility value, i.e. $\sum_{i}x_i ctr_i v_i\leq\mathcal{V}$ for every $x$. The algorithm also uses bounds $V_c$ on the constraint violations.

In \Cref{sec:solutions} we show that there exists an equivalent threshold $T_i$ for the linear program with fairness constraints. In \Cref{sec:algorithm} we write the approximation algorithm for fairness constraints and prove its correctness using the adapted threshold. Using our threshold, the multiplicative weights algorithm can solve the $1$-dimensional problem for the linear program with fairness constraints.

\begin{note}
It is important to note that when seeking integer solutions, certain choices of fairness and budget constraints may be so strict that the only feasible solution is one where no bids are made. This can happen even when fairness constraints would be feasible with an unlimited budget, but are too costly to implement with limited funds. 

In such cases, it would be easy for an autobidding algorithm to notify an advertiser that an inputted constraint set is infeasible. There are many potential ways to relax the budget and/or fairness constraints to achieve a non-trivial feasible solution. However, we want to note that which relaxation an advertiser selects should be given careful consideration as to whether it still aligns with the advertiser's goals and does not disproportionately affect any particular group. What constitutes a ``fair'' relaxation of a constraint set and how to find minimal relaxations with these guarantees is an interesting question for future work. 
\end{note}

\subsection{Solutions to the Linear Program}\label{sec:solutions}
In this section we show that all optimal solutions to LPs of the type described in (\ref{lp:primal}) have a specific structure.

As a first step, we write the linear program and its dual, allowing the solution $x$ to be fractional.

\begin{minipage}{0.4\textwidth}
\begin{align}\label{lp:primal-relaxed}
\begin{array}{ll}
\text{maximize} \sum_{i}x_ictr_iv_i \\
\text{s.t. } \sum_i x_ictr_icpc_i \leq B\\
\sum_i x_i ctr_i (\mu_g - g_i) \leq 0, \forall g \in G\\
0 \leq x_i \leq 1, \quad \forall i \in I
\end{array}
\end{align}
\end{minipage}
\begin{minipage}{0.55\textwidth}
\begin{align}
\begin{array}{ll}
\text{minimize} \sum_{i}\delta_i + \alpha B \\
\text{s.t. } \\
 \delta_i + \alpha ctr_i cpc_i  + \sum_{g} \beta_g ctr_i(\mu_g - g_i) \geq ctr_iv_i \label{eq:dual-constraint}\\
\quad \alpha, \delta_i, \beta_g \geq 0, \quad \forall g \in G, i \in I
\end{array}
\end{align}
\end{minipage}

We show that there is an optimal bidding threshold $T_i$ such that if $x_i^* = 1$ in the optimal solution to the LP above, we have $T_i \geq cpc_i$, and if $x_i^* = 0$, we have $T_i \leq cpc_i$. 

Note that if these inequalities were strict (i.e. $T_i < cpc_i$ and not $\leq$), $T_i$ would provide an optimal bidding formula whose outcomes would match that of the optimal solution. For a second-price auction, the bids would consist of exactly $T_i$, while for a first-price, the advertiser should bid $cpc_i$ 9or $cpc_i+\epsilon$ if this is the winning bid) whenever $T_i > cpc_i$. Because the inequalities are not strict, these thresholds are only used as an intermediate step in the algorithm used to solve the linear program (see Appendix \ref{sec:algorithm}). 
\begin{theorem}\label{thm:solutions}
Let $\mathbf{x}^*$ be the optimal solution to \ref{lp:primal-relaxed}, and for each $i \in I$, let $T_i$ be
\begin{equation}\label{eq:bid-threshold}
T_i := \frac{v_i - \sum_{g \in G} \beta_g (\mu_g - g_i)}{\alpha}.
\end{equation} 
Then, $x_i^* = 0$ implies that $T_i \leq cpc_i$, and $x_i^* = 1$ implies $T_i \geq cpc_i$, with the latter inequality strict whenever $\delta_i > 0$.
\end{theorem}
We prove the theorem via analyzing the complementary slackness conditions of the primal and dual LPs. The proof appears on \Cref{sec:proofs}.

\subsection{Rounding the Solution}\label{sec:rounding}
The solution to this linear program is a vector $x\in [0,1]^n$ that maximizes the objective subject to the given constraints. In this section, we show how to round a fractional solution into an integer solution satisfying the constraints and achieving nearly optimal objective value.

\subsubsection*{Randomized Rounding} One way to interpret the fractional solution $x\in[0,1]^n$ is as a probabilistic solution. That is, for every individual $i\in[n]$, bid $cpc_i$ with probability $x_i$, and else bid $0$. Let $y\in\{0,1\}^n$ be a vector corresponding to a run of this random process, i.e. for every $i$, $y_i\sim\text{Ber}(x_i)$ independently. Let $r_i\in\{0,1\}^n$ be the vector indicating whether an individual clicked on the ad, i.e. for all $i$, $r_i\sim\text{Ber}(ctr_i)$. 
By definition, it means that for all $i$, $\E {y_i}=x_i$ and $\E {r_i}=ctr_i$. 

Since each $r_i$ is a random variable decided by individual $i$, there is an inherent randomness in the outcome and constraint values. Even if we had a deterministic rounding algorithm generating $y$ from $x$, the uncertainty in $r$ does not disappear and we do not get a deterministic expression for the objective and constraints. 
This does not mean that the advertiser would not prefer a stronger guarantee from the solution $y$. For example, the advertiser might want to never exceed the budget.
Given a fractional solution $x$ satisfying certain conditions, we show a randomized rounding algorithm that generates $y$ satisfying all of the constraints with high probability, while only reducing the expected utility by a small factor.
 
 For ease of notation, we say that an ad allocation linear program (\ref{lp:primal-relaxed}) and a solution $x\in[0,1]^n$ are \emph{$\gamma$-flexible} if the set $S_0 = \left\{i\in[n] \left| \sum_{g\in G}g_i=0 \right. \right\}$ satisfies $\sum_{i\in S_0}x_i ctr_i\geq \gamma \sum_{i\in [n]}x_i ctr_i$ and $\sum_{i\in [n]}x_i cpc_i ctr_i\geq\gamma n$.
We remark that if an individual $i$ has $g_i=1$ only for groups $g$ such that $\mu_g=0$, then effectively it is not in any constraint and therefore can be added to $S_0$.

Our rounding algorithm only works for $\gamma$-flexible solutions. We remark that some flexibility in the constraints is required for any rounding algorithm, as can be seen from the following example. Suppose $G = \{g_1,g_2\}$ and that we have two constraints requiring that exactly $1/2$ of the clicks should be from individuals $i\in g_1$ and $1/2$ from $i\in g_2$. Then, because of the inherent randomness in the clicks created by $r_i\sim\text{Ber}(ctr_i)$, it is not possible to promise that both constraints are satisfied with high probability.

\begin{algorithm}
\caption{Randomized rounding algorithm for $\gamma$-flexible linear programm and solution}\label{alg:rand-round}
\KwIn{$x\in[0,1]^n, S_0\subset[n],\epsilon>0$}
\KwOut{$y\in\{0,1\}^n$}
\For {$i=1$ to $n$}
{
$x_i'\gets \begin{cases}
        (1-\epsilon) x_i \quad &i\in S_0\\
        (1-\epsilon/2) x_i \quad &i\notin S_0.
    \end{cases}$\\
$y_i\gets \text{Ber}(x'_i)$
}
\end{algorithm}

\begin{lemma}\label{lem:rand-round}
Let $P$ be an ad allocation linear program, and let $x\in[0,1]^n$ be a fractional solution such that $P,x$ are $\gamma$-flexible and x matches the representation constraint values of the optimal solution to $P$ up to a multiplicative error at most 2. Then for every constant $\epsilon > \gamma$  Algorithm \ref{alg:rand-round} outputs a solution $y\in\{0,1\}^n$, satisfying the following. Let $r$ be the vector representing the individuals clicks, and $\mu = \min_{g\in G}\{\mu_g\}$. Then with probability $1-\exp(-\mu\epsilon^2\gamma^3 n)$ over the randomness of of $y,r$ we have 
\begin{align}
    &\sum_i y_i r_i cpc_i \leq \sum_i x_i ctr_i cpc_i, \label{eq:rand-budget}\\
    &\sum_i y_i r_i (\mu_g-g_i) \leq \sum_i x_i ctr_i  (\mu_g-g_i) \quad \forall g \in G,\label{eq:rand-fairness}\\
     &\E{\sum_i y_i r_i v_i }\geq (1-\epsilon)\sum_i x_i ctr_i v_i. \label{eq:round-obj}
\end{align}
\end{lemma}
The lemma implies that if $x\in[0,1]^n$ satisfies the constraints, then with high probability $y$ satisfies them also. If $x$ approximately satisfies the constraints and has some small error $\delta$, then with high probability $y$ approximately the constraints with the same error. The proof appears on \Cref{sec:proofs}.  
\subsubsection*{Deterministic Rounding}
An interesting variant of our autobidding problem is one in which the advertiser pays for individuals to \emph{view} the ad, rather than clicking on it. This can be modeled by the bid allocation LP in (\ref{lp:primal}) by setting $ctr_i=1$ for every $i\in[n]$. In this setting there is no random variable $r$ representing the clicks and thus no inherent randomness in the outcome. Therefore, we have motivation to discuss a deterministic rounding procedure.

We focus on the special case of disjoint groups, where each individual $i$ has $g(i) = 1$ for exactly one $g \in G$ (some groups might not have constraints). The rounding procedure we present results in a deterministic solution that nearly satisfies all constraints and guarantees approximately optimal utility for the advertiser. Our rounding algorithm works for every solution $x\in [0,1]^n$ satisfying the following condition
\begin{align}\label{eq:round-cond}
    \forall g\in G,i,i'\in g \textbf{ such that } x_i,x_{i'}\in(0,1), \quad v_i > v_{i'} \implies cpc_i > cpc_{i'}.
\end{align}
We remark that from the complimentary slackness, the optimal solution satisfies this condition. Furthermore, for $i,i'\in g$ on which the condition does not hold, $i$ is strictly better than $i'$, so we can increase $x_i$ and reduce $x_{i'}$ and get a better solution. More formally, suppose $x\in[0,1]^n$ is a solution that does not satisfy \cref{eq:round-cond} for some $g$ and $i,i'\in g$, then by changing $x_i$ to $\min\{1,x_i+x_{i'}\}$ and $x_{i'}$ to $\max\{0,x_i+x_{i'}-1\}$ we receive a new solution satisfying all constraints as the original solution, and has at least as good objective. Since checking this condition is efficient, we can easily turn every solution into one satisfying the above without hurting the guarantees. 

\begin{lemma}\label{lem:det-round}
Let $P$ be an ad allocation linear program with disjoint groups $G$ and $ctr_i=1$ for all $i\in[n]$, and let $v_{max} = \max_{i\in I}\{v_i\}$. For every $g\in G$, let $S_g =\left\{i\in[n] \left| g_i=1, x_i\in(0,1) \right. \right\} $. For every fractional solution $x\in[0,1]^n$ satisfying the constraints of $P$ and \cref{eq:round-cond}, Algorithm \ref{alg:det-round} applied on every set $S_g$ outputs a solution $y\in\{0,1\}^n$ such that 
\begin{align}
    &\sum_{i\in[n]} y_i cpc_i \leq \sum_{i\in[n]} x_i cpc_i\leq B \label{eq:det-budget}\\
    &\sum_{i\in[n]} y_i g_i + 1\geq \mu_g \sum_{i\in[n]} y_i \quad \forall g \in G,\label{eq:det-fairness}\\
    &\sum_{i\in[n]} y_i  v_i \geq \sum_{i\in[n]} x_i ctr_i v_i - \abs{G}v_{max}. \label{eq:det-obj}
\end{align}
\end{lemma}

At a high level, the rounding algorithm round down each group separately. That is, if $y\in\{0,1\}^n$ are the rounded values, then for every group $g$ and every value $v$ we have $\sum_{i\in g,v_i\geq v}y_i\leq \sum_{i\in g,v_i\geq v}x_i$.  See the proof on \Cref{sec:proofs} for more details.

\begin{algorithm}
\caption{Deterministic rounding for a single group $S$}\label{alg:det-round}
\KwIn{$S=\{i_1,\ldots,i_t\},x\in[0,1]^n, v\in\mathbb
{R}^n$}
\KwOut{$y\in\{0,1\}^S$}
Assume that the elements in $S$ are ordered according to $v$, i.e. $v_{i_1}\leq v_{i_2}\cdots\leq v_{i_t}$ and in case of equality by $cpc_i$.\\
\For{$j = t$ to $1$}{
   \If {$x_{i_j} + \sum_{l > j}(x_{i_l} - y_{i_l}) \geq 1 $}{
   $y_{i_j}\gets 1$;
   }
   \Else{$y_{i_j}\gets 0$;}
}
\end{algorithm}

We remark that \Cref{alg:det-round} can also be applied in the case of a few not-disjoint set of groups $G$. In this case, we should run it separately over each possible intersection of the groups, i.e. for every $h\in\{0,1\}^{\abs G}$ run is on $S_h = \left\{ i\in[n] \left| \forall g\in G,  h_g = g_i  \right. \right\}$. In this case, instead of violating each constraint by an additive factor of $1$, we have an additive error of $2^{\abs G}$, the loss to the objective value can be $v_{max} 2^{\abs G}$. Therefore, it only makes sense to apply this algorithm for either disjoint, or very few groups $G$.

\subsection{Extension to Online Bidding}\label{sec:online-ext}
Thus far, our autobidder formulation follows prior work which examines an offline setting \cite{aggarwal2019autobidding}. For a large enough advertisement market, generating bids in an offline setting is sufficient due to the high volume and frequency of slots. However, in settings where advertisement slots may be more sparse and there is a fixed time horizon, generating bids that respect budget and representation constraints can be modeled as an online stochastic optimization problem. We assume $cpc_t$, $ctr_t$, $\{g_t\}_{g \in G}$, and $v_t$ are stochastic, meaning that at each time step $t$, a tuple consisting of these values are drawn i.i.d from some stationary distribution. 

We can then define an objective $f_t(x_t)$ and constraints $c_{t, g}^{(0)}, ..., c_{t, g}^{(3)}$ for each group $g \in G$ to give the optimization problem at the $t$th step 
\begin{align*}\label{eq:online} 
f_t(x_t) &= x_tctr_tv_t \\
c_{t, g}^{(0)}(x_t) &= x_t ctr_t (\mu_g - g_t) \leq 0, \quad \forall g \in G\\
c_{t, g}^{(1)}(x_t) &= x_tctr_tcpc_t - \rho \leq 0\\
c_{t, g}^{(2)}(x_t) &= x_t - 1 \\
c^{(3)}_{t, g}(x_t) &= -x_t 
\end{align*}
where $\rho = \frac{B}{T}$ is the goal amount of budget used at every step and $T$ is the time horizon in consideration (i.e. campaign duration).

Using the algorithm for this problem proposed by \cite{castiglioni2022unifying} guarantees an approximate cumulative constraint satisfaction of $\frac{1}{T}\sum_{t=1}^{T}c_{t, g}^{(i)}(x_t) \le \tilde{O}(T^{-1/4})$ for all $i \in [3]$ and $g \in G$. This means that across $T$ steps, our group representation goals can be approximately achieved. Further, this algorithm also gives an upper bound of $\tilde{O}(T^{-1/4})$ on the regret. While \cite{castiglioni2022unifying} also proposed quota-based fairness constraints, they were unable to apply their algorithm because they could not assume the existence of a feasible solution. In constrast, our ratio-based representation constraints always yield a feasible solution: the zeros vector. Moreover, the existence of a strictly feasible solution implies even better guarantees on the cumulative constraint satisfaction and regret.

\section{Experiments}\label{sec:experiments}
To simulate the problem of an employer looking to advertise to a diverse set of candidates, we use data from the US Bureau of Labor Statistics and the American Community Survey. The American Community Survey is a yearly survey given to a sample of the United States population in order to determine how federal and state funds should be distributed. The survey collects information about employment, housing, education, demographic information, and other topics\footnote{\url{https://www.census.gov/programs-surveys/acs/about.html}}. Using 2021 records of individuals in California from this survey~\cite{ding2021retiring}, we construct cost-per-click based on an individual's income and estimate advertiser value by assigning a higher value for individuals in the same occupational category.
To model the higher cost of advertising to women observed by prior works~\cite{lambrecht2019algorithmic}, we add an additional bump uniformly to the cost per click for women such that the average cost-per-click for women is 10\% higher than men. 
We define click-through rates by assuming an individual is more likely to click on an ad if there are more people similar to themselves in the current occupation. This modeling assumption corresponds to stereotype threat~\cite{beasley2012they}; the negative experience caused by being judged based on a negative group stereotype. Using Labor force summary statistics from 2021 \footnote{\url{https://www.bls.gov/cps/cpsaat11.htm}}, we use the gender and race distributions of occupational categories to approximate the click-through rates for an individual query. To account for the variance across income, demographic, and job categories, we add Gaussian noise to value ($v_i$), cost per click ($cpc_i$), and click-through-rate ($ctr_i$), and clip values to a small range. 

\begin{figure}
    \centering
    \includegraphics[width=\textwidth]{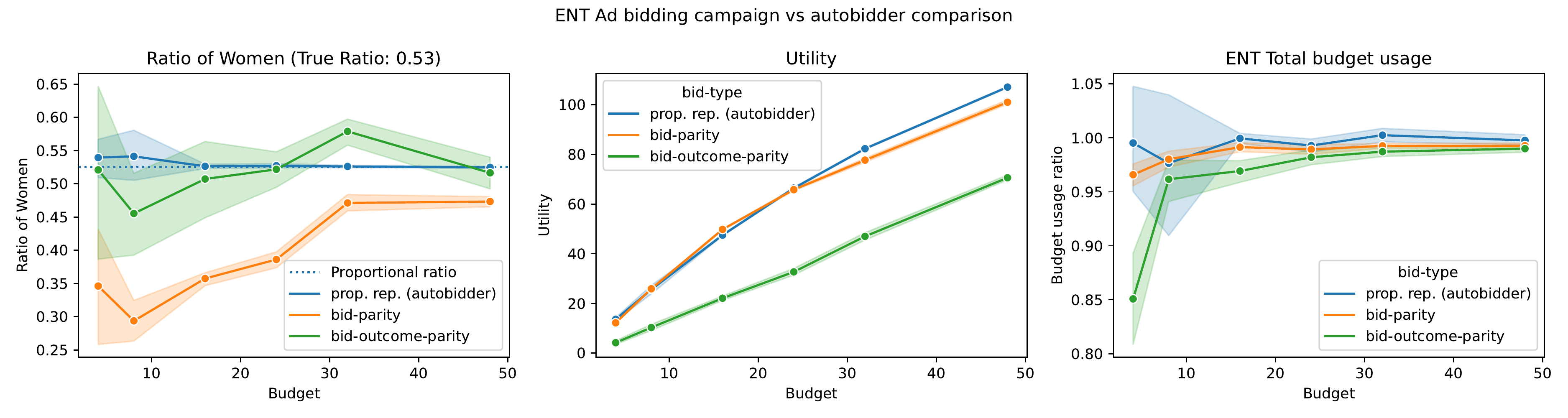}
    \caption{Men and women each represent half of the workforce among \textit{entertainment occupations} workers; we compare the consequences of different fairness objectives in a second price auction. When women cost more to reach than men, using an approach that enforces bid parity guarantees that ads will be shown disproportionately to men; this underrepresentation is particularly stark at a lower budget. Using an auto-bidder with constraints achieves proportional representation while maintaining higher utility than a strategy satisfying bid-and-outcome parity.}
    \label{fig:ent_compare}
\end{figure}

In our experiment setting, we consider a larger pool of viewers both within and outside the target job industry. We set the values of individuals within an industry to be $1.0$ and values for individuals in other occupations to be zero. For each budget, bids are estimated using a disjoint sample from the population that maximizes budget use. We approximate the parity-satisfying bid by finding the $cpc$ threshold in the disjoint population where the budget would become exhausted. For the bid satisfying bid-and-outcome parity, we compare the cumulative distributions of cost-per-click for men and women respectively and find a non-zero intersection point. The cost-per-click at this point reaches a proportional number of men and women. And thus is both bid and outcome fair. While results from previous sections apply to both first and second-price auctions, this set of experiments will be based on second-price auctions. It is easy to see that if we looked at first-price auctions, the bid-parity and bid-and-outcome-parity strategies would be even less efficient in utility with the same budget.     

Figure \ref{fig:ent_compare} compares the bid-parity and bid-and-outcome-parity strategies achieved by a single max bid threshold against our autobidder with proportional group representation constraints in the entertainment industry. This scenario in \textit{entertainment occupations} is motivated by our original example from the introduction, where showing ads to men and women have different costs but men and women appear in the workforce in equal proportion. We see that focusing on bid parity yields a low ratio of women; this effect is especially stark when the total budget is lower. When bid-and-outcome parity is enforced, better representation can be achieved but the utility is strictly lower than the strategy satisfying bid parity. This is because requiring both bid-and-outcome parity results in inefficiency. We apply our autobidder with randomized rounding with parity constraints since parity is equivalent to proportional representation in this industry and plot autobidder candidates for the entertainment industry only. Since the autobidder will use all of the available budget, female candidates not in the entertainment industry may also be selected. Thus while the total number of women candidates is exactly proportional, the number of women in the entertainment industry might be slightly less than proportional. However, our simulations show that the autobidder still achieves representation closer to proportional and yields higher utility than solutions satisfying bid-and-outcome parity. 

\begin{figure}
    \centering
    \includegraphics[width=\textwidth]{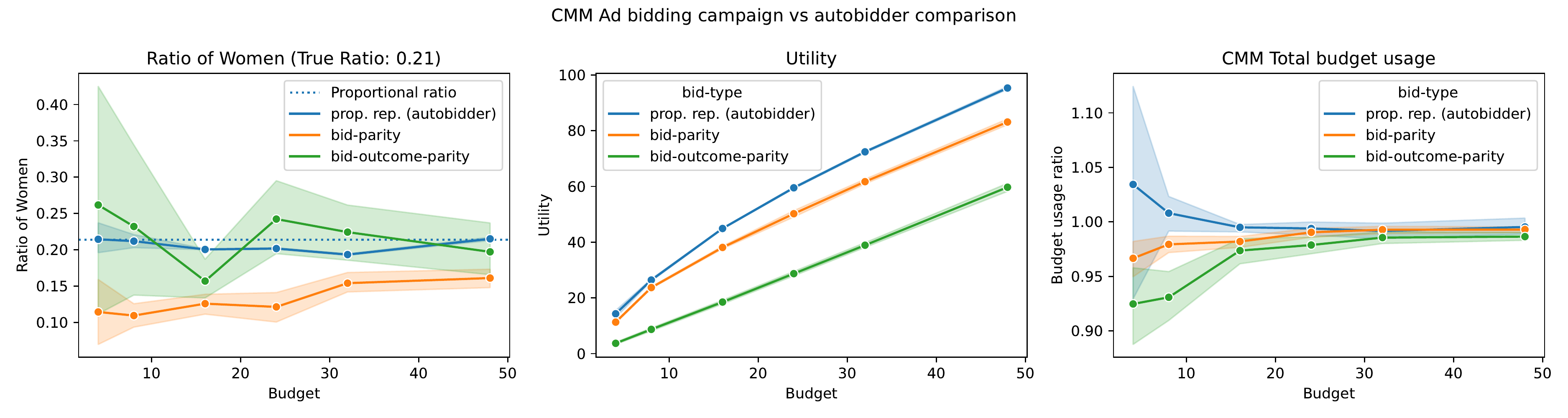}
    \caption{Women only represent 21\% of the workforce in the \textit{computer and mathematical} occupation; we again compare the consequences of various fairness strategies in a second price auction. A bid-parity strategy yields very low female representation but high utility. In contrast, a bid-and-outcome-parity strategy yields proportional representation but lower utility. Meanwhile using an autobidder with proportional constraints yields both good representation and high utility. For a large enough budget, the bid parity and bid-and-outcome parity bids are the same and achieve similar utility and representation.}
    \label{fig:cmm_compare}
\end{figure}

Next, we turn to \textit{computer and mathematical} occupations where women only represent 21\% of workers in our sampled data. Repeating the same process for finding the optimal bid for strategies satisfying bid and bid-and-outcome parity, we can again compare these approaches to our autobidder with proportional representation constraints. Since workers in this industry have much higher incomes, we adjust the minimum cost per click to be slightly higher. In Figure \ref{fig:cmm_compare}, we observe that both our autobidder and the bid-and-outcome-parity strategy achieve better representation than the bid-parity strategy. Comparing utility, we once again observe a significant gap between autobidder and bid-and-outcome-parity utility where employing the autobidder achieves much higher utility. We once again see that the autobidder has higher utility than the bid-parity strategy for the same reason as previously mentioned. Utility-wise, for both occupations, the autobidder always matches or surpasses the bid-parity strategy since some individuals under the threshold may not be the most efficient choices; the autobidder might find a different combination of individuals which maximizes utility that a single threshold cannot achieve. 

In both industries with vastly different baseline demographic compositions, we see that using our autobidder with proportional representation constraints achieves both high levels of representation and utility. For any underrepresented group or intersectional group, we can repeat these examples with similar expected results. If the required level of yield is beyond the population proportion, we can also adjust the target ratio accordingly. 

\section{Discussion}\label{sec:discussion}
Even in the specific setting of group fairness, there are many definitions of fairness and parity that can arise in the advertisement auction and bidding process. We give examples to motivate three potential objectives that have been scattered throughout prior work. 
We discuss what different strategies for achieving each of these goals might look like and give examples of when one notion of fairness (i.e. in bids) might contradict other notions of fairness (e.g. in yield outcome). Our experiments verify the observation from prior work that a strategy satisfying bid parity may result in a lack of diversity when some subgroups are more expensive to advertise to than others. Turning to the bid-and-outcome parity objective, where proportional group representation must be achieved via a bidding strategy that satisfies parity constraints, we show that these additional constraints require much higher bid values to ensure that all populations can be reached. In our simulations, bid-and-outcome parity does achieve better proportional representation than the bid-parity strategy but at the cost of significant utility loss.  

Motivating the case for strategies that satisfy outcome parity, we extend on an existing autobidding framework to include group representation constraints based on the desired ratio of individuals from different groups. Since we use a probabilistic model of cost that is based on click-through rates, we also further modify the autobidder algorithm to satisfy budget and representation constraints with high probability, rather than just in expectation. Incorporating our proposed randomized rounding method that complements our autobidder solution, we show in our experiments that we achieve better outcome fairness than the bid-parity strategy and better utility than the bid-and-outcome-parity strategy. 

In our simplified framework, we assumed that an individual's value to an advertiser can be easily derived based on information about the individual's occupational record. In a real advertising scenario, platforms might have only estimates of viewer employment. Furthermore, there might be systematic biases in terms of missing features like current occupation and income. Designing mechanisms to achieve outcome parity as well a other notions such as individual fairness in the presence of real word data challenges is a promising direction for future work. Furthermore, advertising for job recruitment is just one aspect of recruitment. In reality, a pool of candidates can come from a variety of sources including recruitment events, referrals, job search engines, and direct applications. Each stream of candidates involves different recruitment costs and yield groups with different levels of diversity and skill levels. Exploring composition effects across different sources of recruitment and the underlying network effects that affect which audiences are reached is another interesting direction for future research.

\bibliography{references}
\appendix
\section{Remark on Definitions of Bid Parity}\label{sec:bidfairness}

Throughout this paper, we define a strategy satisfying bid parity as one that selects a single maximum bid $b_{max}$ and bids this value on every member of the target population until the budget runs out. We use this definition because it captures the standard setting in which advertisers can specify their preferences to online advertising platforms by creating a campaign parameterized by a budget, target population, and maximum bid. Moreover, natural relaxations to this strict notion of parity may result in notions that don't guarantee parity with respect to outcomes even in the absence of composition effects. We consider two potential relaxations here to illustrate. 

\subsection{Parity with Respect to Average Bids}

We could imagine a situation in which rather than requiring advertisers bid the same bid with the same probability on all key subgroups, they are instead only required to have the same average bid for each group. 

We show that even in the simplest case where we have two disjoint groups $A$ and $B$ of equal size making up the population and every individual has the same winning bid $w$, only requiring parity with respect to average bids can lead to outcomes where the representation of $A$ and $B$ is far from proportional. 

In particular, consider a strategy that bids $w$ on all individuals from group $A$, while bidding $w - \epsilon$ for some small $\epsilon > 0$ on 90\% of individuals from group $B$, and bidding $w + 9\epsilon$ on the remaining $0.1$\%. For small $\epsilon$ the difference in bids is extremely small, but such a strategy will result in 10x the number of individuals from group $A$ shown ads compared to group B.

\subsection{Approximate Parity}

Similar to above, we might loosen our definition to only require that bids on individuals be close to eachother, i.e. for all individuals $i$ and $j$, we have $|b_i - b_j| < \epsilon$ for some $\epsilon > 0$. 

However, as in our example above, such a constraint can still result in outcomes that are far from proportional even for arbitrarily small values of $\epsilon$. To see how this can occur, consider our example from above where all individuals in $A$ and $B$ have a winning bid of $\epsilon$. One potential strategy in this setting would be to bid $w$ on all individuals from $A$ and $w - \epsilon$ on all individuals from $B$. This results in a strategy that satisfies approximate bid parity constraints, but never shows an ad to an individual from $B$.  

\section{Algorithm for Solving the Linear Program}\label{sec:algorithm}
The bidding algorithm from \cite{aggarwal2019autobidding} can be extended to work with additional group representation constraints. In this section, we explain the bidding algorithm algorithm and prove its correctness when there are additional representation constraints. 

In the algorithm $\delta$ is the approximation parameter, $\mathcal{V}$ is an upper bound on the objective and $V_B,V_G$ are bounds on the value of the budget and group representation constraints. 

At a high level, the algorithm iterates over all possible objective values $V$, and for each value tries to solve the following problem: ``is there an $x$ that satisfies the constraints and has utility $V$?''. This problem can be equivalently restated in matrix form, to ask whether there is an $x$ such that $Ax\geq u$ for the values of $A,u$ described in the algorithm. We use the multiplicative weights algorithm to solve each of these sub-problems. In the update step, the problem is reduced to a problem in $1$-dimension: ``is there an $x$ such that $p^T Ax \geq p^T u$, where $p$ is the weights vector?''. For the $1$-dimensional problem, the optimal threshold described on Section \ref{sec:solutions} is an optimal solution, and therefore can be used for the update.
\begin{algorithm}
\caption{Finding the optimal strategy.}\label{alg:solve}
\KwIn{$\delta>0,\mathcal{V},V_B,V_g \forall g\in G$} 
\KwOut{$\hat x_1,\ldots,\hat x_n\in \{0,1\}$}
\tcc{$\mathcal{V},V_g,V_B,$ are bounds on the objective value and  constraints violations.}
 $T_1 \gets c/\delta$;\\
 $T_2 \gets c/\delta^3$;\\
 $\hat x\gets 0^n$ ;\tcp{output init}
\For{$i=1,\ldots, T_1$} {
$V\gets i\delta\mathcal{V}$;\tcp{$V$ is the current objective we are trying to reach.}
$A \gets \begin{pmatrix}
    ctr_1v_1/\mathcal{V} & ctr_2v_2/\mathcal{V} & \ldots & ctr_nv_n/\mathcal{V} \\
    -ctr_1 cpc_1/V_B & \ldots& \ldots & -ctr_n cpc_n/V_B\\
    ctr_1(g_1 - \mu_g)/V_g &\ldots & \ldots & ctr_n(g_n - \mu_g)/V_g\\
    \vdots & \vdots& \vdots & \vdots
\end{pmatrix}$; 
$u\gets \begin{pmatrix}
    V/\mathcal{V}\\
    -B/V_B \\
    0\\
    \vdots
\end{pmatrix}$;\\
\tcc{MW algorithm solving:  is there $x, 0\leq x_l\leq 1$ such that $Ax\geq u$?} 
$FAIL \gets 0$;\\
$w \gets 1^{2+|G|} $; \tcp{Initialize weights}
\For{$t=1,\ldots,T_2$}{ \tcc{Each iteration solving 1-dim problem:  is there $x, 0\leq x_l\leq 1$ such that $w^T Ax \geq w^T u $?}
$\alpha \gets \frac{w_2\mathcal{V}}{V_B w_1}$;\\
$\beta_g \gets \frac{\mathcal{V}v_B w_g}{w_1 w_2 V_G}$; \tcp{for $g$ the $j$'th group, $w_g = w_{j+2}$}
$b(l) \gets \frac{v_l - \sum_{g\in G}\beta_g(\mu_g - g_l)}{\alpha}$;\\
$x^{(t)}_l \gets \mathbf{1}(b(l)\geq cpc_l) \quad \forall l\in[n]$;\tcp{$x^{(t)}_l$ is the optimal solution to the $1$-dim problem.}
\If {$w^T Ax^{(t)} < w^T u $} {
$FAIL \gets 1$;
}
\Else{
$w_j\gets  \begin{cases}
    w_j\cdot (1-\epsilon)^{A_j x^{(t)} - u_j} \quad &A_j x^{(t)} - u_j\geq 0\\
    w_j\cdot (1+\epsilon)^{-A_j x^{(t)} + u_j} \quad &A_j x^{(t)} - u_j< 0
\end{cases},  \quad \forall j\in[2+|G|]$;
}
}
\If {$FAIL = 0$}{
$\hat x = \sum_{t=1}^{T_2}x^{(t)}$; 
}
}
\end{algorithm}

We state a more formal statement of \Cref{thm:agl} and prove it.
\begin{lemma}\label{lem:solve}
    Let $P$ be a relaxed ad allocation linear program, (\ref{lp:primal-relaxed}). Let $\mathcal{V}$ be an upper bound on the objective value of (\ref{lp:primal-relaxed}) and $V_B,V_G$ be be upper bounds on the amount of violation of the budget and representation constraints. Then for every $\delta>0$, Algorithm \ref{alg:solve} runs in time and $O(n^2/\delta^4|G|)$ and outputs a solution $x\in[0,1]^n$ such that
    \begin{align*}
    &\sum_i x_i ctr_i v_i \geq \textbf{OPT} - \delta\mathcal{V}\\
    &\sum_i x_i ctr_i cpc_i\leq B + \delta V_B\\
    &\sum_i x_i ctr_i  (\mu_g-g_i)\leq \delta V_G \quad \forall g \in G.      
    \end{align*}
\end{lemma}

\begin{proof}
To prove the correctness of the algorithm, it is enough to prove that the $x_l^{(t)}$ assigned is indeed the optimal solution for the $1$-dimensional problem. The rest is implied from the correctness of the multiplicative weight algorithm, see \cite{arora2012multiplicative}.
Therefore, we prove that $x^{(t)}_l$ is the optimal solution to the $1$-dimensional problem $\max_x\{w^TAx-w^Tu\}$.
\begin{align*}
    w^T Ax - w^T u =& w_1\sum_{l=1}^n\frac{v_\ell}{\mathcal{V}}x_l - w_2 \sum_{l=1}^n\frac{ctr_l cpc_l}{V_B}x_l\\&+\sum_{j=3}^{|G|+2}w_j\sum_{l=1}^n \frac{ctr_l(g_l-\mu_g)}{V_g}x_l - w_1 \frac{V}{\mathcal{V}} + w_2\frac{B}{V_B}\\
    =&\sum_{l=1}^n x_l \left(w_1\frac{v_\ell}{\mathcal{V}}-w_2\frac{ctr_l cpc_l}{V_B} + \sum_{j=3}^{|G|+2}w_j \frac{ctr_l(g_l-\mu_g)}{V_g}\right)- w_1 \frac{V}{\mathcal{V}} + w_2\frac{B}{V_B}.
\end{align*}
Denote $C_l = w_1\frac{v_\ell}{\mathcal{V}}-w_2\frac{ctr_l cpc_l}{V_B} + \sum_{j=3}^{|G|+2}w_j \frac{ctr_l(g_l-\mu_g)}{V_g}$. 
The maximal value of $w^T Ax - w^T u $ is given when in every $l$ such that $C_l \geq 0$ we have $x_l=1$, and for the rest we have $x_l=0$.

Notice that after switching $\alpha,\beta_g$ (which we can think about just as renaming of $w_j$) we have that $c_l  \geq 0$ is equivalent to $b(l)\geq cpc_l$, as
\begin{align*}
    C_l =& w_1\frac{ctr_lv_l}{\mathcal{V}} - \frac{w_2}{V_B} ctr_l cpc_l + \sum_g \frac{w_g}{V_g} ctr_l(g_l-\mu_g) \geq 0 \iff\\
    cpc_l\leq& \frac{V_B}{w_2}\left(\frac{ w_1 v_l}{\mathcal{V}}+ \sum_g\frac{w_g}{V_g} (g_l-\mu_g)   \right)
\end{align*}
If we denote $\alpha = \frac{w_2}{V_B}\frac{\mathcal{V}}{w_1}$ and $\beta_g = \frac{\mathcal{V}v_B w_g}{w_1 w_2 V_G}$ then we have that this is the same as $b(l)\geq cpc_l$.
The algorithm solves the $1$-dimensional problem $)(1/\delta^4)$ times, each takes $O(|G|n^2)$ time.
\end{proof}
\section{Proofs of Theorems and Lemmas}\label{sec:proofs}
\begin{proof}[Proof of \Cref{thm:solutions}]
First, suppose we have an $i$ such that $x_i^* = 0$. By the slackness conditions of the LPs, this implies that $\delta_i = 0$. Substituting this fact into constraint \ref{eq:dual-constraint} of the dual tells us that we must have 
\[
\alpha ctr_i cpc_i + \sum_{g \in G}\beta_g ctr_i(\mu_g - g_i) \geq ctr_i v_i
\]

Rearranging the terms of this inequality (and assuming $ctr_i \neq 0$), we get 
\[
cpc_i \geq \frac{v_i - \sum_{g \in G}\beta_g(\mu_g - g_i)}{\alpha}  = T_i
\]

as desired. For the other direction, suppose that $x_i^* = 1$. Again applying complementary slackness, we know that constraint \ref{eq:dual-constraint} must be tight, and thus 
\[
\delta_i + \alpha ctr_i cpc_i + \sum_{g \in G}\beta_g ctr_i(\mu_g - g_i) = ctr_i v_i.
\]
Again rearranging to solve for $cpc_i$, we get:
\[
cpc_i = \frac{v_i - \sum_{g \in G}\beta_g(\mu_g - g_i)}{\alpha} - \frac{\delta_i}{ctr_i \alpha} = T_i - \frac{\delta_i}{ctr_i \alpha}.
\]

We can conclude that this guarantees $T_i \geq cpc_i$, and if $\delta_i > 0$, then $T_i > cpc_i$.
\end{proof}

\begin{proof}[Proof of \Cref{lem:rand-round}]
    Given a $\gamma$-flexible solution $x$, let $y$ be the output of \Cref{alg:rand-round}, and let $x'\in[0,1]^n$ be as in  \Cref{alg:rand-round}. We show that all of the constraints hold with high probability. We denote the realization of clicks from each individual as $r$, i.e.  $r_i\sim\text{Ber}(ctr_i)$.

For the budget constraint, we show that (\ref{eq:rand-budget}) holds with high probability,
\begin{align*}
    \Pr_{y,r}\left[\sum_{i\in [n]}{y_i r_i cpc_i}\geq B\right]&\leq  \Pr_{y,r}\left[\sum_{i\in S}{y_i r_i cpc_i}\geq \left(1+\frac{\epsilon}{2}\right)\sum_{i\in S}x'_i ctr_i cpc_i\right]\leq e^{-\epsilon^2 \gamma^2\left(1-\frac{\epsilon}{2}\right)^2 \frac{n}{4}},
\end{align*}
where the last inequality is due to Hoeffding's inequality.

For the representation constraints, (\ref{eq:rand-fairness}), we have that for every $g\in G$,
\begin{align}\label{eq:g-rep}
    &\Pr_{y_i,r_i}\left[\sum_{i\in [n]}{ g_i y_i r_i }\leq \left(1-\frac{\epsilon}{2}-\frac{\gamma\epsilon}{4}\right)\sum_{i\in[n]}x_i ctr_i g_i\right]\\&~\leq \Pr_{y_i,r_i}\left[\sum_{i\in [n]}{ g_i y_i r_i }\leq \left(1-\frac{\gamma\epsilon}{4}\right)\sum_{i\in[n]}x'_i ctr_i g_i\right]\leq e^{-\frac{\gamma^2\epsilon^2}{32} \sum_{i\in[n]} g_i x_i ctr_i}.
\end{align}
\begin{align}\label{eq:total-rep}
    &\Pr_{y,r}\left[\sum_{i\in [n]}{y_i r_i }\geq \left(1-\gamma\epsilon - (1-\gamma)\frac{\epsilon}{2} +\frac{\gamma\epsilon}{4}\right)\sum_{i\in [n]}{x_i ctr_i }\right] \\&~\leq \Pr_{y,r}\left[\sum_{i\in [n]}{y_i r_i }\geq \left(1+ \frac{\gamma\epsilon}{4}\right)\sum_{i\in [n]}{x'_i ctr_i }\right]\leq e^{-\frac{\gamma^2\epsilon^2}{32} \sum_{i\in[n]}x_i ctr_i}.
\end{align}
The solution $x$ satisfies the constraint up to a constant error of $2$, so  $\sum_{i\in[n]} g_i x_i ctr_i \geq 1/2\cdot \mu_G \sum_{i\in[n]} x_i ctr_i$. Therefore the bound in both (\ref{eq:g-rep}) and (\ref{eq:total-rep}) is at most $exp(-\gamma^3\epsilon^2\mu_g n)$.
If the events in (\ref{eq:g-rep}) and (\ref{eq:total-rep}) do not hold, then the representation constraint on group $g$ is satisfied, as we have that 
\begin{align*}
    &\mu_g\sum_{i\in[n]} y_i r_i \leq \left(1-\frac{\epsilon}{2}-\frac{\gamma\epsilon}{4}\right)\mu_g \sum_{i\in[n]} x_i ctr_i,\quad 
    \sum_{i\in[n]}g_i y_i r_i \geq \left(1-\frac{\epsilon}{2}-\frac{\gamma\epsilon}{4}\right)\sum_{i\in[n]}{g_i x_i ctr_i}
\end{align*}

By union bound over all group representation constraints for $g\in G$ and over the budget constraint, with probability $1-\exp(-\mu \gamma^2 \epsilon^2 n)$ all constraints hold.

We are left with showing that the objective is not reduced by much. We notice that $\forall i\in[n], x_i'\geq(1-\epsilon)x_i$, so from the linearity of expectation we get (\ref{eq:round-obj}).

\end{proof}

\begin{proof}[Proof of \Cref{lem:det-round}]
    Let $S =S_g$ for some $g\in G$.  Let $i_1,\ldots,i_t$ be the order of the elements in $S$ used by the algorithm. From (\ref{eq:round-cond}), this order is also an order by $cpc_{i}$.

    From the algorithm, we have that for every $j\in[t]$,
    \begin{align}
       \sum_{l\geq j}x_{i_l}-1\leq \sum_{l\geq j}y_{i_l}\leq \sum_{l\geq j}x_{i_l}.\label{eq:round-b}
    \end{align}
    
    For the budget constraint, (\ref{eq:det-budget}), 
    we claim that for every $j\in[t]$,
    \begin{align}
        \sum_{l\geq j}cpc_{i_l}(x_{i_l}-y_{i_l}) \geq cpc_{i_j}\sum_{l\geq j}(x_{i_l}-y_{i_l}).\label{eq:neg-price}
    \end{align}
    We prove it by induction on $j$, starting from $j=t$. The basis is implied from (\ref{eq:round-b}). The step,
    \begin{align*}
    \sum_{l\geq j}cpc_{i_l}(x_{i_l}-y_{i_l}) =& cpc_{i_j}(x_{i_j}-y_{i_j}) + \sum_{l >  j}cpc_{i_l}(x_{i_l}-y_{i_l})  \\\geq& 
    cpc_{i_l}(x_{i_l}-y_{i_l}) + cpc_{i_{j+1}}\sum_{l >  j}(x_{i_l}-y_{i_l}) \tag{Inductive step}\\ \geq& 
    cpc_{i_l}(x_{i_l}-y_{i_l}) + cpc_{i_{j}}\sum_{l >  j}(x_{i_l}-y_{i_l}).
    \end{align*}
    Where in the last inequality we use the facts that $cpc_{i_{j}}\leq cpc_{i_{j+1}}$ and $\sum_{l >  j}(x_{i_l}-y_{i_l})\geq 0$.
    Applying (\ref{eq:neg-price}) with $j=1$ and using (\ref{eq:round-b}) implies that $\sum_{i\in S}y_i  cpc_i \leq \sum_{i\in S}x_i  cpc_i$, and in general $\sum_{i\in [n]}y_i  cpc_i \leq \sum_{i\in [n]}x_i  cpc_i$, proving (\ref{eq:det-budget}).

    For the representation constraint, we have from (\ref{eq:round-b}) that for every $g\in G$, $\sum_{i\in[n]}y_i g_i \geq \sum_{i\in[n]}x_i g_i -1 $.
    By summing up on all $S$, we get that $\sum_{i\in[n]}y_i \leq \sum_{i\in[n]}x_i.$
    Together with the fact that $x$ satisfy the representation constraint we get
    \[ \sum_{i\in[n]}y_i g_i +1 \geq \sum_{i\in[n]}x_i g_i \geq \mu_g \sum_{i\in[n]}x_i \geq \mu_g \sum_{i\in[n]}y_i . \]
    Therefore,  $y$ satisfy (\ref{eq:det-fairness}) for every group $g$.

   For the objective value, (\ref{eq:det-obj}), we fix a set $S$ and let $i_1,\ldots, i_t$, be the order used in the algorithm. To simplify the proof, we ``split'' elements in $S$ and divide their $x_i$ in the following way: if we have $y_{i_j}=1$ because $x_{i_j} + \sum_{l > j}(x_{i_l} - y_{i_l}) > 1 $, then we split $i_j$ to two elements $i,i'$ with $x_i = 1-\sum_{l > j}(x_{i_l} - y_{i_l})$ and $x_{i'} = x_{i_j}-x_i$. This ``splitting'' is for analysis only, and we abuse notation by denoting $S=\{i_1,\ldots i_t\}$ also after the splitting. After the splitting we have that if $y_{i_j}=1$ then $\sum_{l\geq j}x_l = \sum_{l\geq j} y_l$. 
    
     Let $j_1,\ldots j_k\in [t]$ be the indices in which $y_j=1$. We have that for every $m\in [k]$, $\sum_{l=j_m}^{j_{m+1}-1}x_l=1$, and also $\sum_{l\geq j_k}x_l=1$ and $\sum_{l< j_1}x_l<1$.
    Therefore,
    \begin{align*}
        \sum_{l\in[t]}v_{i_l}x_{i_l} = &\sum_{l=1}^{j_1-1} v_{i_l}x_{i_l} + \sum_{l=j_1}^{j_2-1} v_{i_l}x_{i_l} + \cdots+ \sum_{l=j_m}^{t} v_{i_l}x_{i_l}\\ \leq&
        v_{i_{j_1}}\sum_{l=1}^{j_1-1}x_{i_l} + v_{i_{j_2}}\sum_{l=j_1}^{j_2-1} x_{i_l} + \cdots +v_{i_t}\sum_{l=j_m}^{t}x_{i_l} \tag{$v_i$ are increasing}\\ \leq&
        v_{i_{j_1}} + v_{i_{j_2}} + \cdots + v_{i_t} \leq v_{i_t}+ \sum_{l\in[t]}y_{i_l} v_{i_l},
    \end{align*}
    which proves (\ref{eq:det-obj}).
\end{proof}

\end{document}